\documentclass[a4paper,UKenglish]{lipics-v2019}

\usepackage[utf8]{inputenc}
\usepackage{cmll}
\usepackage{verbatim}
\usepackage{bussproofs}
\usepackage{tikz}
\usepackage{tikz-cd}
\usepackage{csquotes}
\usepackage{xcolor}
\usepackage{xspace}
\usepackage{multicol}
\usepackage[normalem]{ ulem }
\usepackage[all,2cell]{xy}\SelectTips{cm}{10}\SilentMatrices\UseAllTwocells  
\usepackage{mathtools}
\mathtoolsset{centercolon}

\usepackage{accents}

\theoremstyle{plain}

\theoremstyle{definition}

\newcommand\LCb{\LC_{\beta\eta}} 
\newcommand\LCbfix{\LC_{\beta\eta,\fix}}

\newcommand{\pullback}{\ar@{}[dr]|<<{\ulcorner}}
\newcommand{\pushout}{\ar@{}[ul]|<<{\ulcorner}}

\newcommand{\Mon}{\ensuremath{\mathsf{Mon}}}
\newcommand{\Mod}{\ensuremath{\mathsf{Mod}}}

\newcommand{\LMod}{\ensuremath{\int_R\Mod(R)}}

\newcommand{\WRep}{\ensuremath{\int_{\Sigma}\Mon^{\Sigma}}}

\newcommand{\rar}{\longrightarrow}
\newcommand{\LC}{{\mathsf{LC}}}

\newcommand{\app}{{\mathsf{app}}}
\newcommand{\abs}{{\mathsf{abs}}}

\newcommand{\fix}{{\mathsf{fix}}}

\newcommand{\Cat}[1]{\mathsf{#1}}

\newcommand{\Set}{{\Cat{Set}}}

\newcommand{\swap}{{\mathsf{swap}}}

\newcommand{\NN}{{\mathbb{N}}}

\newcommand{\oneSig}{\text{1-}\mathsf{Sig}}

\newcommand{\Sig}{\ensuremath{\mathsf{Sig}}}
\newcommand{\twoSig}{\ensuremath{2\Sig}}
\newcommand{\twoMod}{\ensuremath{2\Mod}}
\newcommand{\InttwoMod}{\ensuremath{\int_{(\Sigma,E)} \Mon^{(\Sigma,E)}}}

\newcommand{\inl}{\ensuremath{\mathsf{inl}}}
\newcommand{\inr}{\ensuremath{\mathsf{inr}}}

\newcommand{\m}{\mathsf{m}}
\newcommand{\e}{\mathsf{e}}

\usepackage{xifthen}

\newcommand{\UniMath}{\href{https://github.com/UniMath/UniMath}{\nolinkurl{UniMath}}\xspace}

\newcommand{\codebaseurl}{https://github.com/UniMath/largecatmodules}

\newcommand{\coqlonghash}{1539d1c399784e54e28b13b52d2d4f7e27369b8e}
\newcommand{\coqhash}{1539d1c}
\newcommand{\coqdocbasebaseurl}{https://initialsemantics.github.io/doc/\coqhash}
\newcommand{\coqdocbaseurl}{\coqdocbasebaseurl /Modules.}
\newcommand{\urlhash}{\#}
\newcommand{\coqdocurlhtml}[1]{\coqdocbaseurl #1.html}
\newcommand{\coqdocurl}[2]{\coqdocurlhtml{#1}\urlhash #2}
\newcommand{\nolinkcoqident}[1]{\nolinkurl{#1}} 
\makeatletter
\newcommand{\coqident}{\begingroup\@makeother\#\@coqident}
\newcommand{\@coqident}[3][]{%
  \ifthenelse{\isempty{#2}}%
  {\nolinkcoqident{#3}}%
  {\ifthenelse{\isempty{#1}}%
  {\href{\coqdocurl{#2}{#3}}{\nolinkcoqident{#3}}}%
  {\href{\coqdocurl{#2}{#3}}{\nolinkcoqident{#1}}}}%
\endgroup}

\makeatother

\title{Modular specification of monads through higher-order presentations}

\author{Benedikt Ahrens}
       {University of Birmingham, United Kingdom}
       {B.Ahrens@cs.bham.ac.uk}
       {https://orcid.org/0000-0002-6786-4538}
       {Ahrens acknowledges the support of the Centre for Advanced Study (CAS) in Oslo, Norway, which funded and hosted the research project \emph{Homotopy Type Theory and Univalent Foundations} during the 2018/19 academic year.}
\author{Andr\'e Hirschowitz}
       {Université Nice Sophia Antipolis, France}
       {ah@unice.fr}
       {https://orcid.org/0000-0003-2523-1481}
       {}
\author{Ambroise Lafont}
       {IMT Atlantique
	\\
        Inria, LS2N CNRS, France}
       {ambroise.lafont@inria.fr}
       {https://orcid.org/0000-0002-9299-641X}{}
\author{Marco Maggesi}
       {Università degli Studi di Firenze, Italy}
       {marco.maggesi@unifi.it}
       {https://orcid.org/0000-0003-4380-7691}
       {Supported by GNSAGA-INdAM and MIUR.}
\authorrunning{B. Ahrens, A. Hirschowitz, A. Lafont, and M. Maggesi}

\Copyright{Benedikt Ahrens, Andr\'e Hirschowitz, Marco Maggesi, Ambroise Lafont}



\relatedversion{}

\supplement{Computer-checked proofs with compilation instructions on \url{\codebaseurl /tree/\coqhash}}

\funding{This work has partly been funded by the CoqHoTT ERC Grant 637339.
This material is based upon work supported by the Air Force Office of 
Scientific Research under
award number FA9550-17-1-0363.}

\acknowledgements{We thank Paige R.~North for a valuable hint regarding preservation of epimorphisms.}

\nolinenumbers 
\hideLIPIcs  

\begin{document}
\maketitle

\begin{abstract}
In their work on second-order equational logic,
Fiore and Hur have studied presentations of simply typed languages by generating binding constructions and equations among them.
To each pair consisting of a binding signature and a set of equations, they associate a category of \enquote{models}, 
and they give a monadicity result which implies that this category has an initial object, which is the language presented by the pair.

In the present work, we propose, for the untyped setting, a variant of their approach where monads and modules over them are the central notions.
More precisely, we study, for monads over sets, presentations by generating (`higher-order') operations and equations among them.
We consider a notion of 2-signature which allows to specify a \emph{monad} with a family of binding operations subject to a family of equations, 
as is the case for the paradigmatic example of the lambda calculus, specified by its two standard constructions (application and abstraction) subject to $\beta$- and $\eta$-equalities.
Such a 2-signature is hence a pair ($\Sigma$,E) of a binding signature $\Sigma$ and a family $E$ of equations for $\Sigma$.
This notion of 2-signature has been introduced earlier by Ahrens in a slightly different context. 

We associate, to each 2-signature $(\Sigma,E)$, a category of \enquote{models of $(\Sigma,E)$}; and we say that a 2-signature is \enquote{effective} if this category has an initial object;
the monad underlying this (essentially unique) object is the \enquote{monad specified by the 2-signature}.
Not every 2-signature is effective; we identify a class of 2-signatures, which we call \enquote{algebraic}, that are effective.

Importantly, our 2-signatures together with their models enjoy \enquote{modularity}:
when we glue (algebraic) 2-signatures together, their initial models are glued accordingly.

We provide a computer formalization for our main results.

\end{abstract}

\section{Introduction}
The present work is devoted to the study of presentations of monads on the category of sets.  More precisely, there is a well established theory of presentations of monads through generating (first-order) operations equipped with relations among the corresponding derived operations.
Here we propose a counterpart of this theory, where we consider generation of monads by \emph{binding operations}. 
Various algebraic structures generated by binding operations have been considered by many,
going back at least to Fiore, Plotkin, and Turi \cite{FPT}, Gabbay and Pitts \cite{gabbay_pitts99}, and Hofmann \cite{hofmann}.
Every such operation has a binding arity, which is a sequence of non-negative integers. 
For example,
the binding arity of the application operation of the lambda calculus is
$(0,0)$: it takes two arguments without binding any variable in them, while
 the abstraction operation on the monad of the lambda calculus
 has binding arity $(1)$,
as it binds one variable in its single argument. 
For each family $\Sigma$ of binding arities, there is a generated \enquote{free} monad $\hat{\Sigma}$ on $\Set$ which maps a set of free variables $X$ to the set of terms $\hat{\Sigma}(X)$ taking variables in $X$.

If $p : \hat{\Sigma} \to R$ is a monad epimorphism, we understand that $R$ is generated by a family of operations whose binding arities are given by $\Sigma$, subject to suitable identifications. 
In particular, for $\Sigma := ((0,0), (1))$, $\hat{\Sigma}$ may be understood as the monad $\LC$ of syntactic terms of the lambda calculus, and we have an obvious epimorphism $p : \hat{\Sigma} \to \LCb$, where $\LCb$  is the monad of lambda-terms modulo $\beta$ and $\eta$.
In order to manage such equalities, the approach in the first-order case suggests to 
identify $p$ as the coequalizer of a double arrow from $T$ to $\hat{\Sigma}$ where $T$ is again a \enquote{free} monad. 
Let us see what comes out when we attempt to find such an encoding for the $\beta$-equality of the monad $\LCb$. It should say that for each set $X$, the following two maps from $\hat{\Sigma}(X+\{*\})\times\hat{\Sigma}(X)$ to $\hat{\Sigma}(X)$,
\begin{itemize}
    \item $(t,u) \mapsto \app(\abs(t),u)$
    \item $(t,u)\mapsto t[*\mapsto u]$
\end{itemize}
are equal. Here a problem occurs, namely that the above collections of maps, which can be understood as a morphism of functors, cannot be understood as a morphism of monads. Notably, they do not send variables to variables.  

On the other hand, we observe that the members of our equations, which are not morphisms of monads, commute with substitution, and hence are more than morphisms of functors: indeed they are morphisms of \emph{modules over $\hat{\Sigma}$}. (In Section~\ref{sec:modules}, we recall briefly what modules over a monad are.) 
Accordingly, a (second-order)  presentation for a monad $R$ could be a diagram
\begin{equation}
\xymatrix{
 T \ar@<-.5ex>[r] \ar@<.5ex>[r]^f & \hat{\Sigma} \ar[r]^p & R
 }
 \label{eq:intro-coeq-diag}
\end{equation}
where $\Sigma$ is a binding signature, $\hat{\Sigma}$ is the associated free monad, $T$ is a module over $\hat{\Sigma}$, $f$ is a pair of morphisms of modules over $\hat{\Sigma}$, and $p$ is a monad epimorphism. 
And now we are faced with the task of finding a condition meaning something like \enquote{$p$ is the coequalizer of $f$}\footnote{This cannot be the case stricto sensu since $f$ is a pair of morphisms of modules while $p$ is a morphism of monads.}. 

 To this end, we introduce the category $\Mon^\Sigma$ \enquote{of models of $\Sigma$}, whose objects are monads \enquote{equipped with an action of $\Sigma$}.
Of course $\hat{\Sigma}$ is equipped with such an action which turns it into the 
initial object. Next, we define the full subcategory  of
models satisfying the equation $f$, and require $R$ to be the initial object
therein. Our definition is suited for the case where the equation $f$ is
parametric in the model: this means that now $T$ and $f$ are functions of the
model $S$, and $f(S)=(u(S) , v(S))$ is a pair of $S$-module morphisms from $T(S)$ to $S$.
We say that $S$ satisfies the equation $f$ if $u(S)=v(S)$. 
Generalizing the case of one equation to the case of a family of equations yields the notion of 2-signature already introduced by Ahrens \cite{ahrens_relmonads} in a slightly different context.

Now we are ready to formulate our main problem: given a 2-signature
$(\Sigma,E)$, where $E$ is a family of parametric equations as above,
does the subcategory of models of $\Sigma$ satisfying the family of equations $E$ admit an initial object?

We answer positively for a large subclass of 2-signatures which
we call \emph{algebraic} 2-signatures (see Theorem~\ref{thm:alg-elem-2-sigs-are-representable}).

This provides a construction of a monad from an algebraic 2-signature, and  we prove furthermore (see Theorem~\ref{thm:modularity}) that this construction is \emph{modular}, in the sense that merging two extensions of 2-signatures corresponds to building an amalgamated sum of initial models. This is 
analogous to our previous result for 1-signatures  shown in \cite[Thm.~32]{ahrens_et_al:LIPIcs:2018:9671}.

As expected, our initiality property generates a recursion principle which is a recipe allowing us
to specify a morphism from the presented monad to any given other monad.

We give various examples of monads arising \enquote{in nature} that can be specified 
via an algebraic 2-signature (see Section~\ref{sec:examples-higher-order-theories}), and we also show through a simple example how our recursion principle applies (see Section~\ref{sec:recursion}).

\subparagraph*{Computer-checked formalization}

 This work is accompanied by a computer-checked formalization of the main results,
 based on 
 the formalization of a previous work \cite{ahrens_et_al:LIPIcs:2018:9671}.
 We work over the \textsf{UniMath} library \cite{UniMath},
 which is implemented in the proof assistant \textsf{Coq} \cite{Coq:manual}.
 The formalization consists of about 9,500 lines of code, and can be consulted on \url{\codebaseurl}.
A guide is given in the \href{\codebaseurl /blob/\coqhash/README.md}{README}, and 
a summary of our formalization is available at \url{\coqdocurlhtml{SoftEquations.Summary}}.

For the purpose of this article, we refer to a fixed version of our library,
with the short hash \href{\codebaseurl /tree/\coqlonghash}{\coqhash}.
This version compiles with version \href{https://github.com/UniMath/UniMath/tree/b16841738c2593f66834eb8f0ca79b4b667d4ed6}{b168417}
of \UniMath. 

Throughout the article, statements are accompanied by their corresponding identifiers in the formalization.
These identifiers are also hyperlinks to the online documentation stored at \url{\coqdocbasebaseurl/index.html}.

\subparagraph*{Related work}
The present work follows a previous work of ours
\cite{ahrens_et_al:LIPIcs:2018:9671} where we study
a slightly different kind of presentation of monads.
Specifically, in \cite{ahrens_et_al:LIPIcs:2018:9671}, we treat a class
of 1-signatures which can be understood as quotients of algebraic 1-signatures.
This should amount to considering a specific kind of equations, as suggested
in Section~\ref{s:col-alg-sigs}, where we recover, in the current setting,
all the examples given there.

In \cite{ahrens_et_al:LIPIcs:2018:9671}, we discussed related work on the general topic of monads and syntax.
Let us focus here on related work on presentations of 
languages (or monads).

In an abstract setting, \cite{KP} explains how any finitary monad can be presented as a coequalizer of free monads. There, free monads correspond to our initial models of an algebraic 1-signature without any binding construction.

In \cite{DBLP:conf/csl/FioreH10}, the authors introduce a notion of equation based on syntax with \emph{meta-variables}:
essentially, a specific syntax, say, $T := T(M,X)$ considered there depends on two contexts: a meta-context $M$, and an object-context $X$.
The terms of the actual syntax are then those terms $t \in T(\emptyset, X)$ in an empty meta-context.
An equation for $T$ is, simply speaking, a pair of terms in the same pair of contexts. Transferring an equation to any model of the 
underlying algebraic 1-signature  is done by induction on the 
syntax with meta-variables. The authors show a monadicity theorem which straightforwardly implies an initiality result very
similar to ours.

As said before, Ahrens \cite{ahrens_relmonads} introduces the notion of
2-signature which we consider here,  in the slightly different context of
(relative) monads on preordered sets, where the preorder models the reduction
relation. In some sense, our result tackles the technical issue of quotienting
the initial (relative) monad constructed in \cite{ahrens_relmonads} by the preorder.

\section{Categories of modules over monads}
\label{sec:modules}

In this section, we recall the notions of monad and module over a monad,
as well as some constructions of modules. 
We restrict our attention to the category $\Set$ of sets, although most definitions are straightforwardly generalizable.
See \cite{Hirschowitz-Maggesi-2010} for a more extensive introduction.

A \textbf{monad} (over $\Set$) is a triple
$R=(R,\mu,\eta)$ given by a functor $R\colon\Set \rar
\Set$, and two natural transformations $\mu\colon R\cdot R \rar R$ and
$\eta\colon I \rar R$ such that the well-known monadic laws hold.
A \textbf{monad morphism} to another such monad $(R',\mu',\eta')$ is a natural transformation
$f : R \to R'$ that commutes with the monadic structure. 
The category of monads is denoted by $\Mon$.

Let $R$ be a monad. A \textbf{(left) $R$-module}%
\footnote{The analogous notion of \emph{right} $R$-module is not used in this work, we hence
simply write \enquote{$R$-module} instead of \enquote{left $R$-module} for brevity.}
is given by a functor $M\colon \Set \rar \Set$
  equipped with a natural transformation $\rho\colon M \cdot R \rar
  M$, called \emph{module substitution}, which is compatible with the monad
  composition and identity:
  \begin{equation*}
    \rho \circ \rho R = \rho \circ M \mu, \qquad
    \rho \circ M\eta = 1_M.
  \end{equation*}

\noindent
 Let $f\colon R\rar S$ be a morphism of monads and $M$ an $S$-module.
  The module substitution
  \begin{math}
    M\cdot R \stackrel{Mf}\rar M \cdot S \stackrel\rho\rar M
  \end{math}
  turns $M$ into an $R$-module $f^*\!M$, called \textbf{pullback of $M$ along $f$}.

A natural transformation of $R$-modules $\varphi\colon M \rar N$ is
  \textbf{linear} if it is compatible with module substitution on either side,
  that is, if $\varphi \circ \rho^M = \rho^N \circ \varphi R$.
Modules over $R$ and their morphisms form a category denoted $\Mod(R)$,
which is complete and cocomplete: limits and colimits are computed pointwise.

We define the \textbf{total module category} $\LMod$ as follows:
its objects are pairs $(R, M)$ of a monad $R$ and an $R$-module $M$.
A morphism from $(R, M)$ to $(S, N)$ is a pair $(f, m)$ where
    $f\colon R \rar S$ is a morphism of monads, and $m\colon M \rar
    f^*N$ is a morphism of $R$-modules.  
  The category $\LMod$ comes
    equipped with a forgetful functor to the category of monads, given
    by the projection $(R,M) \mapsto R$. 
    This functor is a Grothendieck fibration with fibers $\Mod(R)$ over $R$.
In particular, any monad morphism $f : R \rar S$ gives rise to a functor
  $f^*\colon \Mod(S) \rar \Mod(R)$
  which preserves limits and colimits.

We give some important examples of modules:
\begin{example}
  \label{ex:modules}
  \begin{enumerate}
  \item Every monad $R$ is a module over itself, which we call the
    \textbf{tautological module}.
  \item For any functor $F\colon \Set \rar \Set$ and any $R$-module $M\colon
    \Set \rar \Set$, the composition $F\cdot M$ is an $R$-module (in the evident
    way).
  \item For every set $W$ we denote by $\underbar W \colon \Set \rar \Set$ the
    constant functor $\underbar W := X \mapsto W$. Then $\underbar W$ is
    trivially an $R$-module since $\underbar W = \underbar W \cdot R$.
  \item Given an $R$-module $M$, the $R$-module $M'$ is defined, on objects, as
    $M'(X) := M(X+\{*\})$, and the obvious module substitution. Derivation
    yields an endofunctor on $\Mod(R)$ that is right adjoint to the functor $M
    \mapsto M\times R$, \enquote{product with the tautological module}.
    Details are given, e.g., in \cite[Sec.\ 2.3]{ahrens_et_al:LIPIcs:2018:9671}.
    \label{ex:deriv}
  \item Derivation can be iterated. Given a list of non negative integers
    $(a)=(a_1,\dots,a_n)$ and a left module $M$ over a monad $R$, we denote by
    $M^{(a)}=M^{(a_1,\dots,a_n)}$ the module $M^{(a_1)}\times\cdots\times
    M^{(a_n)}$, with $M^{()}=1$ the final module.
  \end{enumerate}
\end{example}

\section{1-signatures and their models}
\label{sec:1-signatures}

In this section, we review the notion of 1-signature studied in detail in 
\cite{ahrens_et_al:LIPIcs:2018:9671} --- there only called \enquote{signature}.

  A \textbf{1-signature} is a section of the forgetful functor
  from the category $\LMod$ to the category $\Mon$.
   A \textbf{morphism between two 1-signatures} $\Sigma_1, \Sigma_2 \colon \Mon \rar \LMod$
  is a natural transformation $m \colon \Sigma_1 \rar \Sigma_2$ which, post-composed
  with the projection $\LMod \rar \Mon$, is the identity.
  The category of 1-signatures is denoted by $\oneSig$.

  Limits and colimits of 1-signatures can be easily constructed pointwise:
  the category of 1-signatures is complete and cocomplete.

Table~\ref{table:1-sigs} lists important examples of 1-signatures.
\begin{table}[htb]
\begin{tabular}{lll}
Hypotheses & On objects & Name of the 1-signature \\ \hline 
&$R \mapsto R$ & $\Theta$\\
$\Sigma$ 1-signature, $F$ functor &$R \mapsto F \cdot \Sigma(R)$ & $F \cdot \Sigma$\\
&$R \mapsto 1_R$ & $1$\\
$\Sigma$, $\Psi$ 1-signatures &$R \mapsto \Sigma(R)\times \Psi(R)$ & $\Sigma \times \Psi$\\
$\Sigma$, $\Psi$ 1-signatures &$R \mapsto \Sigma(R) + \Psi(R)$ & $\Sigma + \Psi$ \\
                                    &$R \mapsto R'$ &       $\Theta'$\\
$n \in \NN$                       &$R \mapsto R^{(n)}$ &  $\Theta^{(n)}$ \\
$(a)=(a_1,\dots,a_n) \in \NN^n$   &$R \mapsto R^{(a)} = R^{(a_1)} \times \ldots \times R^{(a_n)}$ &  $\Theta^{(a)}$ \text{\textbf{elementary signatures}}
\end{tabular}
\caption{Examples of 1-signatures}
\label{table:1-sigs}
\end{table}
   An \textbf{algebraic 1-signature} is a (possibly infinite) coproduct of elementary signatures (defined in Table~\ref{table:1-sigs}).
  For instance, the algebraic 1-signature of the lambda calculus is $\Sigma_\LC = \Theta^2 + \Theta'$.

  Given a monad $R$ over $\Set$, we define an
  \textbf{action of the 1-signature $\Sigma$ in $R$} to be a module morphism from
  $\Sigma(R)$ to $R$.
  For example, the application $\app\colon \LC^2 \rar \LC$ is an action of the elementary 1-signature
  $\Theta^2$ into the monad $\LC$ of syntactic lambda calculus.
  The abstraction $\abs\colon \LC' \rar \LC$ is an action of the elementary 1-signature
  $\Theta'$ into the monad $\LC$. 
  Then $[\app, \abs] : \LC^2 + \LC' \rar \LC$ is an action of the algebraic 1-signature
  of the lambda-calculus $\Theta^2 + \Theta'$ into the monad $\LC$.

  Given a  1-signature $\Sigma$, we build the category
  $\Mon^{\Sigma}$ of \textbf{models of $\Sigma$} as follows. Its
  objects are pairs $(R,r)$ of a monad $R$ equipped with an
  action $r : \Sigma(R) \to R$ of $\Sigma$.  A morphism  from $(R, r)$ to $(S, s)$ is a
  morphism of monads $m : R \to S$ making the following diagram of
  $R$-modules commutes:
  \[
    \xymatrix{
      **[l] \Sigma(R) \ar[r]^{r}\ar[d]_{\Sigma(m)} & **[r] R \ar[d]^m \\
      **[l] m^* (\Sigma(S)) \ar[r]_{m^*s} & **[r] m^* S}
  \]

  Let $f\colon \Sigma \rar \Psi$ be a morphism of 1-signatures
     and $\mathcal{R}=(R,r)$ a model of $\Psi$.
  The linear morphism
  \begin{math}
    \Sigma(R)\stackrel{f(R)}\rar\Psi(R)\stackrel{r}\rar R
  \end{math}
  defines an action of $\Sigma$ in $R$. The induced model of $\Sigma$ is called
  \textbf{pullback}
   of $\mathcal{R}$ along
  $f$ and noted $f^*\!\mathcal{R}$.

The \textbf{total category} $\WRep$ of models is defined as follows:
\begin{itemize}
\item An object of $\WRep$ is a triple $(\Sigma, R, r)$ where $\Sigma$ is a 1-signature, 
 $R$ is a monad, and $r$ is an action of $\Sigma$
  in $R$.
\item A morphism in $\WRep$ from $(\Sigma_1, R_1, r_1)$ to $(\Sigma_2, R_2, r_2)$ 
  consists of a pair $(i,m)$ of a 1-signature morphism $i: \Sigma_1 \rar \Sigma_2$
  and a morphism $m$ of $\Sigma_1$-models from $(R_1, r_1)$
  to $(R_2, i^*(r_2))$.
\end{itemize}
The forgetful functor $\WRep \to \Sig$ is a Grothendieck fibration.

Given a 1-signature $\Sigma$,
 the initial object in $\Mon^\Sigma$, if it exists, is denoted by
 $\hat{\Sigma}$.
  In this case, 
   the 1-signature $\Sigma$ is said
   \textbf{effective}%
    \footnote{In our previous work~\cite{ahrens_et_al:LIPIcs:2018:9671},
    we call \textbf{representable} any 1-signature $\Sigma$ that has an initial
    model,
    called a \textbf{representation} of $\Sigma$, or \textbf{syntax generated
    by} $\Sigma$.}.

\begin{theorem}[{\cite[Theorems 1 and 2]{HM}}]
  \label{t:alg-sig-representable}
  Algebraic 1-signatures are effective.
\end{theorem}

\section{2-Signatures and their models}
\label{sec:2-sigs}

In this section we study \emph{2-signatures} and \emph{models of 2-signatures}.
A 2-signature is a pair of a 1-signature and a family of \emph{equations} over it.

\subsection{Equations}

Our equations are those of Ahrens \cite{ahrens_relmonads}, namely they are 
parallel module morphisms parametrized by the models of the underlying 1-signature.
The underlying notion of 1-model is essentially the same as in \cite{ahrens_relmonads}, even if, there, such equations are interpreted instead as \emph{inequalities}.

Throughout this subsection, we fix a 1-signature $\Sigma$, that we instantiate in the examples.

\begin{definition}
  We define a \textbf{$\Sigma$-module} to be a functor $T$ from the category of
  models of $\Sigma$ to the category $\LMod$ commuting with
  the forgetful functors to the category $\Mon$ of monads,
 \[
  \begin{xy}
   \xymatrix{
                    \Mon^{\Sigma} \ar[rd] \ar[rr]^{T}  &        &   \LMod \ar[dl] \
                    \\
                                     & \Mon
   }
  \end{xy}
 \]

\end{definition}

\begin{example}
\label{ex:sigma-mod-from-1-sig}

To each 1-signature $\Psi$ is associated, by precomposition with the
projection from $\Mon^\Sigma$ to $\Mon$, a $\Sigma$-module still
denoted $\Psi$.
All the $\Sigma$-modules occurring in this work arise in this way from 1-signatures;
in other words, they do not depend on the action of the 1-model.
In particular, we have the 
\textbf{tautological 
$\Sigma$-module $\Theta$}, and,
more generally, for any natural number $n \in \NN$, a
$\Sigma$-module $\Theta^{(n)}$. Also we have another fundamental $\Sigma$-module (arising in this way from) $\Sigma$ itself. 

\end{example}

\begin{definition}
  Let $S$ and $T$ be $\Sigma$-modules.
  We define a \textbf{morphism of $\Sigma$-modules} from $S$ to $T$ to be a natural
  transformation from $S$ to $T$ which becomes the identity when postcomposed with the
  forgetful functor from the category of models of $\Sigma$ to the category of monads. 
\end{definition}  

\begin{example}
  \label{ex:sigma-mor-of-sig-mor}
  Each 1-signature morphism $\Psi\to\Phi$ upgrades into
  a morphism of $\Sigma$-modules.
  Further in that vein, there is a
  morphism of $\Sigma$-modules  $\tau^{\Sigma} : \Sigma \to \Theta$.
  It is given, on a model $(R,m)$ of $\Sigma$, by 
  $m : \Sigma(R) \to R$. (Note that it does not arise from a morphism of 1-signatures.)
  When the context is clear, we write simply $\tau$ for this morphism, and call it the \textbf{tautological morphism of $\Sigma$-modules}.
\end{example}

\begin{proposition}
Our $\Sigma$-modules and their morphisms, with the obvious composition and identity,
 form a category.
\end{proposition}

\begin{definition}
  We define a $\Sigma$-equation 
  to be a pair of
  parallel morphisms of $\Sigma$-modules.  We also
  write $e_1=e_2$ for the $\Sigma$-equation $e = (e_1, e_2)$.
\end{definition}

\begin{example}[Commutativity of a binary operation]\label{ex:half-equation-mult-swap}
Here we instantiate our fixed 1-signature as follows: $\Sigma := \Theta \times \Theta$. In this case, we say that $\tau$ is the (tautological) binary operation.
Now we can formulate the usual law of commutativity for this binary operation.
 
 We consider the morphism of 1-signatures $\swap : \Theta^2 \rar \Theta^2$ that exchanges the two components of the direct product.
 Again by Example~\ref{ex:sigma-mor-of-sig-mor}, we have an induced morphism of $\Sigma$-modules, still denoted $\swap$.
 
    Then, the $\Sigma$-equation for commutativity is
    given by the two morphisms of $\Sigma$-modules
    \begin{equation*}
      \xymatrix@R=2pt{
        \Theta^2 \ar[r]^{\swap} &
        \Theta^2 \ar[r]^-{\tau} &
        \Theta \\
        \Theta^2 \ar[rr]_{\tau} &
        {} &
        \Theta}
    \end{equation*}
    See also Section~\ref{sec:example-monoids} where we explain in detail the case of monoids.
\end{example}

For the example of the lambda calculus with $\beta$- and $\eta$-equality (given in Example~\ref{ex:lcbe}),
we need to introduce \emph{currying}:

\begin{definition}
  \label{def:half-eq-curry}
  By abstracting over the base monad $R$
  the adjunction in the category of $R$-modules of Example~\ref{ex:modules},
  item~\ref{ex:deriv}, we can perform \textbf{currying} of morphisms of 1-signatures: given a morphism of signatures $\Sigma_1 \times \Theta \to \Sigma_2$ it produces a new morphism $\Sigma_1 \to \Sigma_2'$.
  By Example~\ref{ex:sigma-mod-from-1-sig}, currying acts also on morphisms of $\Sigma$-modules.

  Conversely, given  a morphism of 1-signatures (resp. $\Sigma$-modules)
  $\Sigma_1 \to \Sigma_2'$, we can define the \textbf{uncurryied} map $\Sigma_1
  \times \Theta \to \Sigma_2$.
  \end{definition}

\begin{example}[$\beta$- and $\eta$-conversions]\label{ex:lcbe}
  Here we instantiate our fixed 1-signature as follows: 
  $\Sigma_\LC := \Theta\times \Theta + \Theta'$. This is the 
  1-signature
  of the lambda calculus. We break the tautological $\Sigma$-module morphism into its two pieces, namely $\app := \tau \circ \inl : \Theta\times \Theta \rar \Theta$ and 
  $\abs := \tau \circ \inr : \Theta' \rar \Theta $.
  Applying currying to
  $\app$
  yields the morphism 
  $\app_1 : \Theta \rar \Theta'$
  of $\Sigma_\LC$-modules.
  The usual $\beta$ and $\eta$ relations are implemented in our
  formalism by two $\Sigma_\LC$-equations that we call $e_\beta$ and $e_\eta$ respectively:
  \begin{equation*}
    e_\beta :
    \vcenter{\vbox{
  \xymatrix@R=2pt{
    \Theta' \ar[r]^{\abs} &\Theta \ar[r]^{\app_1} &\Theta' \\
    \Theta' \ar[rr]_{1} && \Theta'}}}
  \qquad\text{and}\qquad
  e_\eta :
  \vcenter{\vbox{
  \xymatrix@R=2pt{
    \Theta \ar[r]^{\app_1} &\Theta' \ar[r]^{\abs} &\Theta \\
    \Theta \ar[rr]_{1} && \Theta}}}
\end{equation*}
\end{example}

\subsection{2-signatures and their models}

\begin{definition}
  A \textbf{2-signature} is a pair $(\Sigma,E)$ of a 1-signature $\Sigma$ and a
  family $E$ of $\Sigma$-equations. 
\end{definition}

\begin{example}\label{ex:2-sig-comm-mult}

  The 2-signature for a commutative binary operation is 
  $(\Theta^2, \tau \circ \swap = \tau)$ (cf.\ Example~\ref{ex:half-equation-mult-swap}).
\end{example}

\begin{example}\label{ex:2-sig-lcbe}
  The 2-signature of the lambda calculus modulo $\beta$- and $\eta$-equality is 
  $\Upsilon_{\LCb}=(\Theta \times \Theta + \Theta', \{e_\beta,e_\eta\})$,
  where $e_\beta,e_\eta$ are the $\Sigma_\LC$-equations defined in Example~\ref{ex:lcbe}.
\end{example}

\begin{definition}[\coqident{SoftEquations.Equation}{satisfies_equation}]
  We say that a model $M$ of $\Sigma$ 
  \textbf{satisfies the $\Sigma$-equation $e = (e_1, e_2$)} if $e_{1}(M) = e_{2}(M)$.
  If $E$ is a family of $\Sigma$-equations, we say that a
  model $M$ of $\Sigma$ \textbf{satisfies $E$} if $M$ satisfies each
  $\Sigma$-equation in $E$.
\end{definition}
\begin{definition}
  Given a monad $R$ and a 2-signature $\Upsilon=(\Sigma,E)$, an \textbf{action
  of $\Upsilon$ in $R$} is an action of $\Sigma$ in $R$ such that the induced
  1-model satisfies all the equations in $E$.
\end{definition}

\begin{definition}[\coqident{SoftEquations.Equation}{precategory_model_equations}]
  For a 2-signature
  $(\Sigma,E)$, we define the \textbf{category $\Mon^{(\Sigma,E)}$ of
  models of $(\Sigma, E)$} to be the full subcategory of the
  category of models of $\Sigma$ whose objects are
  models of $\Sigma$ satisfying $E$, or equivalently, monads 
  equipped with an action of $(\Sigma , E)$.
\end{definition}

\begin{example}
\label{ex:2sig-lcbeta}
 A model of the 2-signature $\Upsilon_{\LCb}=(\Theta\times\Theta + \Theta',\{e_\beta,e_\eta\})$ is given by a model
 $(R,\app^R : R\times R \to R, \abs^R : R' \to R)$ of the 1-signature $\Sigma_\LC$
 such that $\app^R_1 \cdot \abs^R = 1_{R'}$ and $\abs^R \cdot \app^R_1
 = 1_R$ (see Example~\ref{ex:lcbe}).
\end{example}

\begin{definition}
  A 2-signature $(\Sigma,E)$ is said to be
  \textbf{effective}
  if its category of models $\Mon^{(\Sigma,E)}$ has an initial object,
  denoted $\widehat{(\Sigma,E)}$.
\end{definition}

In Section~\ref{sec:initiality}, we aim to find sufficient conditions for a
2-signature $(\Sigma,E)$ to be effective.

\subsection{Modularity for 2-signatures}

In this section, we define the category $\twoSig$ of 2-signatures and
the category $\twoMod$ of models of 2-signatures, together with
functors that relate them with the categories of 1-signatures and
1-models.
The situation is summarized in the commutative diagram of functors
  \[
   \begin{xy}
        \xymatrix@C=4pc@R=5pc{        
                   **[l] \twoMod \rtwocell<4>_{F_\Mod}^{U_\Mod}{'\top} \ar[d]_{2\pi} & **[r] \Mod \ar[d]^{\pi}
                \\
                   **[l] \twoSig \rtwocell<4>_{F_\Sig}^{U_\Sig}{'\top} & **[r] \Sig 
                   }
       \end{xy} \enspace 
\]
where
 \begin{itemize}
  \item $2\pi$ is a Grothendieck fibration;
  \item $\pi$ is the Grothendieck fibration defined in \cite[Section~5.2]{ahrens_et_al:LIPIcs:2018:9671};
  \item $U_\Sig$ is a coreflection and preserves colimits; and
  \item $U_\Mod$ is a coreflection.
 \end{itemize}
As a simple consequence of this data, we obtain a \emph{modularity} result in Theorem~\ref{thm:modularity}:
it explains how the initial model of an amalgamated sum of 2-signatures is the 
amalgamation of the initial model of the summands.

We start by defining the category $\twoSig$ of 2-signatures:
\begin{definition}[\coqident{SoftEquations.CatOfTwoSignatures}{TwoSignature_category}]
 Given 2-signatures $(\Sigma_1,E_1)$ and $(\Sigma_2,E_2)$, a \textbf{morphism of 2-signatures from $(\Sigma_1,E_1)$ to $(\Sigma_2,E_2)$} is
 a morphism of 1-signatures $m : \Sigma_1 \to \Sigma_2$ such that
 for any model $M$ of $\Sigma_2$ satisfying $E_2$, the $\Sigma_1$-model $m^*M$ satisfies $E_1$.

 These morphisms, together with composition and identity inherited from 1-signatures, form the category $\twoSig$.
\end{definition}

We now study the existence of colimits in $\twoSig$. 
We know that $\Sig$ is cocomplete, and we use this knowledge
in our study of $\twoSig$, by relating the two categories:

Let $F_\Sig : \Sig \to \twoSig$ be the functor which associates to any 1-signature $\Sigma$ the empty family of equations,
$F_\Sig(\Sigma) := (\Sigma, \emptyset)$.
Call $U_\Sig : \twoSig \to \Sig$ the forgetful functor defined on objects as $U(\Sigma, E) := \Sigma$.
\begin{lemma}
 [\coqident{SoftEquations.CatOfTwoSignatures}{TwoSignature_To_One_right_adjoint}, \coqident{SoftEquations.CatOfTwoSignatures}{OneSig_to_TwoSig_fully_faithful}]
 The forgetful functor $U_\Sig$ is a coreflection and is right adjoint to $F_\Sig$.
\end{lemma}

We are interested in specifying new languages by \enquote{gluing together}
simpler ones.
On the level of 2-signatures, this is done by taking the coproduct, or, more generally,
the pushout of 2-signatures:

\begin{theorem}
 [\coqident{SoftEquations.CatOfTwoSignatures}{TwoSignature_PushoutsSET}]
 \label{thm:2sig-pushouts}
 The category $\twoSig$ has
  pushouts.
\end{theorem}

Coproducts are computed by taking the union of the equations and the coproducts of the underlying 1-signatures.
Coequalizers are computed by keeping the equations of the codomain and taking the
coequalizer of the underlying 1-signatures.  Thus, by decomposing any colimit
into coequalizers and coproducts, we have this more general result:
\begin{proposition}
  The category $\twoSig$  is cocomplete and $U_\Sig$ preserves colimits.
\end{proposition}

We now turn to our modularity result, which states that
the initial model of a coproduct of two 2-signatures is the coproduct of the 
initial models of each 2-signature. 
More generally, the two languages can be amalgamated along a common
\enquote{core language}, by considering a pushout rather than a coproduct.

For a precise statement of that result, we define a \enquote{total category of models of 2-signatures}:
\begin{definition}
\label{d:2mod}
 The category $\InttwoMod$, or $\twoMod$ for short, has, as objects, pairs $((\Sigma, E), M)$ of a 2-signature $(\Sigma,E)$ and
  a model $M$ of $(\Sigma,E)$.

 A morphism from $((\Sigma_1, E_1), M_1)$ to $((\Sigma_2, E_2), M_2)$ is a pair $(m,f)$ consisting of
  a morphism $m : (\Sigma_1, E_1) \to (\Sigma_2, E_2)$ of 2-signatures and
  a morphism $f : M_1 \to m^* M_2$ of $(\Sigma_1,E_1)$-models (or, equivalently, of $\Sigma_1$-models).
\end{definition}

This category of models of 2-signatures contains the models of 1-signatures as a coreflective subcategory. 
Let $F_\Mod : \Mod \to \twoMod$ be the functor which associates to any 
1-model $(\Sigma , M)$ the empty family of equations,
$F_\Mod(\Sigma,M) := (F_\Sig(\Sigma), M)$.
Conversely, the forgetful functor $U_\Mod : \twoMod \to \Mod$ 
maps $((\Sigma,E),M)$ to $(\Sigma,M)$.

\begin{lemma}
 [\coqident{SoftEquations.CatOfTwoSignatures}{TwoMod_To_One_right_adjoint}, \coqident{SoftEquations.CatOfTwoSignatures}{OneMod_to_TwoMod_fully_faithful}]
 We have $F_\Mod \dashv U_\Mod$. Furthermore, $U_\Mod$ is a coreflection.
\end{lemma}

The modularity result is a consequence of the following technical result:

\begin{proposition}
 [\coqident{SoftEquations.CatOfTwoSignatures}{two_mod_cleaving}]
 The forgetful functor $2\pi:\twoMod \to \twoSig$ is a Grothendieck fibration.
\end{proposition}

The \emph{modularity result} below is analogous to the modularity result for 1-signatures \cite[Thm.\ 32]{ahrens_et_al:LIPIcs:2018:9671}:

\begin{theorem}[Modularity for 2-signatures,
\coqident{SoftEquations.Modularity}{pushout_in_big_rep}]
\label{thm:modularity}
 Suppose we have a pushout diagram of effective 2-signatures, as on the left below.
 This pushout gives rise to a commutative square of morphisms of models in
 $\twoMod$ as on the right below,
 where we only write the second components, omitting the (morphisms of) signatures.
 This square is a pushout square.
 \[
  \xymatrix{
    \Upsilon_0 \ar[r]\ar[d] & \Upsilon_1 \ar[d]\\
    \Upsilon_2 \ar[r] & \Upsilon \pushout
    }
   \qquad
   \xymatrix{
    \widehat{\Upsilon}_0 \ar[r]\ar[d] & \widehat{\Upsilon}_1 \ar[d]\\
    \widehat{\Upsilon}_2 \ar[r] & \widehat{\Upsilon} \pushout
    }
\]
\end{theorem}
Intuitively, the 2-signatures $\Upsilon_1$ and $\Upsilon_2$ specify two extensions of 
the 2-signature $\Upsilon_0$, and $\Upsilon$ is the smallest extension containing both these extensions.
By Theorem~\ref{thm:modularity} the initial model of $\Upsilon$
is the \enquote{smallest model containing both
the languages generated by $\Upsilon_1$ and $\Upsilon_2$}.

\subsection{Initial Semantics for 2-Signatures}
\label{sec:initiality}

We now turn to the problem of constructing the initial model of a
2-signature $(\Sigma,E)$.
More specifically, we identify sufficient conditions for $(\Sigma,E)$
to admit an initial object $\widehat{(\Sigma,E)}$ in
the category of models.
Our approach is very straightforward: we seek to construct $\widehat{(\Sigma,E)}$ 
by applying a suitable quotient construction to the initial object $\hat\Sigma$ of $\Mon^\Sigma$.

This leads immediately to our first requirement on $(\Sigma,E)$, 
which is that $\Sigma$ must be an effective 1-signature.  
(For instance, we can assume that $\Sigma$ is an algebraic 1-signature, see Theorem \ref{t:alg-sig-representable}.)  
This is a very natural hypothesis, since in the case where $E$ is the empty
family of $\Sigma$-equations, it is obviously a necessary and
sufficient condition.

Some $\Sigma$-equations are never satisfied.  In
that case, the category $\Mon^{(\Sigma,E)}$ is empty.
  For example, given any 1-signature $\Sigma$, consider the $\Sigma$-equation 
  $\inl,\inr : \Theta \rightrightarrows \Theta + \Theta$ given by the 
  left and right inclusion.
  This is obviously an unsatisfiable $\Sigma$-equation.
We have to find suitable hypotheses to rule out such unsatisfiable
$\Sigma$-equations.  This  motivates the notion of \emph{elementary} equations.

\begin{definition}
  Given a 1-signature $\Sigma$, a $\Sigma$-module $S$ is \textbf{nice}
   if $S$ sends pointwise epimorphic
  $\Sigma$-model morphisms to pointwise epimorphic module morphisms.
\end{definition}

\begin{definition}[\coqident{SoftEquations.quotientequation}{elementary_equation}]
  Given a 1-signature $\Sigma$,
  an \textbf{elementary $\Sigma$-equation} is a $\Sigma$-equation such that 
  \begin{itemize}
    \item the target is a finite derivative of 
          the tautological 2-signature $\Theta$, i.e., of the form $\Theta^{(n)}$ for some $n \in \NN$, and 
   \item the source is a nice $\Sigma$-module.
  \end{itemize}
\end{definition}
\begin{example}
  \label{ex:elementary-algebraic}
  Any algebraic 1-signature is nice
  \cite[Example 43]{ahrens_et_al:LIPIcs:2018:9671}.
  Thus, any $\Sigma$-equation between an algebraic 1-signature and $\Theta^{(n)}$, for some
  natural number $n$, is elementary.
\end{example}

\begin{definition}
  A 2-signature $(\Sigma,E)$ is said \textbf{algebraic}
  if
  $\Sigma$ is algebraic and $E$ is a family of
  elementary equations.
\end{definition}

\begin{theorem}
 [\coqident{SoftEquations.AdjunctionEquationRep}{elementary_equations_on_alg_preserve_initiality}]
 \label{thm:alg-elem-2-sigs-are-representable}
 Any algebraic 2-signature has an initial model.
\end{theorem}
The proof of Theorem~\ref{thm:alg-elem-2-sigs-are-representable} is given in Section~\ref{secappendix:proof-main}.

\begin{example}
\label{ex:lcbe2}
The 2-signature of lambda calculus modulo $\beta$ and $\eta$ equations given in Example~\ref{ex:2-sig-lcbe} is algebraic.
Its initial model is precisely the monad $\LCb$ of lambda calculus modulo $\beta\eta$ equations.

The instantiation of the formalized Theorem~\ref{thm:alg-elem-2-sigs-are-representable} to this 2-signature is done in \coqident{SoftEquations.Examples.LCBetaEta}{LCBetaEta}%
\footnote{%
An initiality result for this particular case was also previously discussed and proved formally in the Coq proof assistant in \cite{Hirschowitz-Maggesi-2010}.}.

\end{example}

Let us mention finally that, using the axiom of choice, we can take a similar quotient 
 on all the 1-models of $\Sigma$:
\begin{proposition}
  [\coqident{SoftEquations.AdjunctionEquationRep}{forget_2model_is_right_adjoint},
  \coqident{SoftEquations.Equation}{forget_2model_fully_faithful}]
  Here we assume the axiom of choice. 
 The forgetful functor from the category $\Mon^{(\Sigma,E)}$ of 2-models 
 of $(\Sigma,E)$ to the category $\Mon^{\Sigma}$ of $\Sigma$-models has a left adjoint.
 Moreover, the left adjoint is a reflector.
\end{proposition}

\section{Proof of Theorem~\ref{thm:alg-elem-2-sigs-are-representable}}
\label{secappendix:proof-main}

Our main technical result on effectiveness is the following Lemma~\ref{lem:hoat-initial}.
In Theorem~\ref{thm:alg-elem-2-sigs-are-representable}, we give a much simpler criterion that encompasses all the examples we give.

\begin{lemma}[\coqident{SoftEquations.AdjunctionEquationRep}{elementary_equations_preserve_initiality}]
  \label{lem:hoat-initial}
  Let $(\Sigma,E)$ be a 2-signature such that:
  \begin{enumerate}
    \item $\Sigma$ sends epimorphic natural transformations to epimorphic natural transformations,
      \label{thm-cond-epimon}
    \item $E$ is a family of elementary equations,
      \label{thm-cond-softeq}
    \item the initial 1-model of $\Sigma$ exists,
      \label{thm-cond-ini}
    \item the initial 1-model  of $\Sigma$ preserves epimorphisms,
      \label{thm-cond-iniepi}
    \item the image by $\Sigma$ of the initial 1-model of $\Sigma$ preserves epimorphisms.
      \label{thm-cond-epipw}
  \end{enumerate}
  Then, the category of 2-models of $(\Sigma,E)$ has an initial object.
\end{lemma}

Before tackling the proof of Lemma~\ref{lem:hoat-initial}, we 
discuss how to derive Theorem~\ref{thm:alg-elem-2-sigs-are-representable} from it, and we prove some auxiliary results.

The \enquote{epimorphism} hypotheses of Lemma~\ref{lem:hoat-initial} are used to transfer structure
from the initial model $\hat{\Sigma}$ of the 1-signature $\Sigma$ onto a suitable quotient. 
There are different ways to prove these hypotheses:
\begin{itemize}
\item The axiom of choice implies conditions \ref{thm-cond-iniepi} and \ref{thm-cond-epipw} since,
  in this case, any epimorphism in $\Set$ is split and thus preserved by any functor.
\item Condition $\ref{thm-cond-epipw}$ is a consequence of condition \ref{thm-cond-iniepi}
  if $\Sigma$ sends monads preserving epimorphisms to modules preserving epimorphisms.
\item If $\Sigma$ is algebraic, then conditions
  \ref{thm-cond-epimon}, 
  \ref{thm-cond-ini}, \ref{thm-cond-iniepi} and \ref{thm-cond-epipw}
  are satisfied \cite[Example 43 and Lemma 45]{ahrens_et_al:LIPIcs:2018:9671}.
\end{itemize}

\noindent
From the remarks above, we derive the simpler and weaker statement of Theorem~\ref{thm:alg-elem-2-sigs-are-representable} that covers
all our examples, which are algebraic.

This section is dedicated to the proof of the main technical result,
Lemma~\ref{lem:hoat-initial}.
The reader inclined to do so may safely skip this section, and rely
on the correctness of the machine-checked proof instead.

The proof of Lemma~\ref{lem:hoat-initial} uses some quotient constructions that we present now:

\begin{proposition}
  [\coqident{Prelims.quotientmonadslice}{u_monad_def}]
  \label{prop:quotient-monad}
  Given a monad $R$ preserving epimorphisms and a collection of monad morphisms $(f_i:R \rightarrow S_i)_{i\in I}$,
  there exists a \emph{quotient} monad $R/(f_i)$ together with a \emph{projection}
  $p^R\colon R \rar R/(f_i)$, which is a morphism of monads such that each $f_i$ factors
  through $p$.
\end{proposition}
\begin{proof}
  The set $R/(f_i)(X)$ is computed as the quotient of $R(X)$  with respect to the
  relation $x \sim y$ if and only if $f_i(x)=f_i(y)$ for each $i\in I$.
  This is a straightforward adaptation of Lemma 47 of
  \cite{ahrens_et_al:LIPIcs:2018:9671}. 
\end{proof}

Note that the epimorphism preservation is implied by the axiom of choice, but
can be proven for the monad underlying the initial model $\hat{\Sigma}$ of an algebraic 1-signature $\Sigma$ 
even without resorting to the axiom of choice.

The above construction can be transported on $\Sigma$-models:

\begin{proposition}
[\coqident{SoftEquations.quotientrepslice}{u_rep_def}]
  \label{prop:soft-quotient}
  Let $\Sigma$ be a 1-signature sending epimorphic natural transformations to epimorphic
  natural transformations, and let $R$ be a $\Sigma$-model such that $R$ and $\Sigma(R)$ preserve epimorphisms.
  Let $(f_i : R \rightarrow S_i)_{i\in I}$ be a collection of $\Sigma$-model morphisms.  Then the monad
  $R/(f_i)$ has a natural structure of $\Sigma$-model and the
  quotient map $p^R\colon R \rar R/(f_i)$ is a morphism of
  $\Sigma$-models. Any morphism $f_i$ factors through $p^R$ in the category of
  $\Sigma$-models.
\end{proposition}
The fact that $R$ and $\Sigma(R)$ preserve epimorphisms is implied by the axiom of choice.
The proof follows the same line of reasoning as the proof of Proposition~\ref{prop:quotient-monad}.

Now we are ready to prove the main technical lemma:

\begin{proof}[Proof of Lemma~\ref{lem:hoat-initial}]

Let $\Sigma$ be an effective 1-signature, and let $E$ be a set of elementary $\Sigma$-equations.
The plan of the proof is as follows:

\begin{enumerate}
 \item Start with the initial model $(\hat\Sigma,\sigma)$,
  with $\sigma : \Sigma (\hat\Sigma) \to \hat\Sigma$.
\item Construct the quotient model $\hat\Sigma/(f_i)$ according to Proposition \ref{prop:soft-quotient} where
  $(f_i : \hat\Sigma \rightarrow S_i)_i$ is the collection of all initial $\Sigma$-morphisms from $\hat\Sigma$ to
  any $\Sigma$-model satisfying the equations.
  We denote by $\sigma/(f_i) : \Sigma(\hat\Sigma/(f_i)) \rightarrow \hat\Sigma/(f_i)$ the action
  of the quotient model.
 \item Given a model $M$ of the 2-signature $(\Sigma, E)$, we obtain a morphism $i_M : \hat\Sigma/(f_i) \to M$ from
   Proposition \ref{prop:soft-quotient}.  Uniqueness of $i_M$ is shown
   using epimorphicity of the projection $p:\hat\Sigma \rightarrow \hat\Sigma/(f_i)$.
   For this, it suffices to show uniqueness of the composition $i_M \circ p : \hat{\Sigma} \to M$
   in the category of 1-models of $\Sigma$,
   which follows from initiality of $\hat{\Sigma}$.
 \item The verification that $\bigl(\hat\Sigma/(f_i), \sigma/(f_i)\bigr)$ satisfies the equations is given below.
   Actually, it follows the same line of reasoning as in the proof of Proposition \ref{prop:quotient-monad}
   that $\hat\Sigma/(f_i)$ satisfies the monad equations.
\end{enumerate}
%
%
Let $e = (e_1, e_2) : U \to \Theta^{(n)}$ be an elementary equation of $E$.
We want to prove that the two arrows
  \begin{equation*}
    e_{1,\hat\Sigma/(f_i)}, e_{2,\hat\Sigma/(f_i)} \colon
    U(\hat\Sigma/(f_i))\rar (\hat\Sigma/(f_i))^{(n)}
  \end{equation*}
  are equal.
  As $p$ is an epimorphic natural transformation, $U(p)$ also is by definition of an elementary equation.
  It is thus sufficient to prove that 
  \[e_{1,\hat\Sigma/(f_i)} \circ U(p) =  e_{2,\hat\Sigma/(f_i)} \circ U(p) \enspace , \]
  which, by naturality of $e_1$ and $e_2$, is equivalent to
  $p^{(n)}\circ e_{1,\hat\Sigma} = p^{(n)}\circ e_{2,\hat\Sigma}$.

  Let $x$ be an element of $U(\hat\Sigma)$ and let us show that 
  $p^{(n)}( e_{1,\hat\Sigma} (x)) = p^{(n)}( e_{2,\hat\Sigma}(x))$.
  By definition of $\hat\Sigma/(f_i)$ as a pointwise quotient (see Proposition \ref{prop:quotient-monad}),
  it is enough to show that for any $j$, the equality 
  $f_j^{(n)}( e_{1,\hat\Sigma} (x)) = f_j^{(n)}( e_{2,\hat\Sigma}(x))$ is satisfied.
  Now, by naturality of  $e_1$ and $e_2$, this equation is equivalent to 
  $ e_{1,S_j}(U(f_j)(x))) = e_{2,S_j}(U(f_j)(x)))$
  which is true since $S_j$ satisfies the equation $e_1 = e_2$.
  \end{proof}

\section{Examples of algebraic 2-signatures}
\label{sec:examples-higher-order-theories}

We already illustrated our theory by looking at the paradigmatic case of lambda calculus modulo $\beta$- and $\eta$-equations (Examples~\ref{ex:lcbe} and~\ref{ex:lcbe2}).  This section collects further examples of application of our results.

In our framework, complex signatures can be built out of simpler ones by taking their
coproducts. 
Note that the class of algebraic 2-signatures encompasses the algebraic 1-signatures and
is closed under arbitrary coproducts:
the prototypical examples of algebraic 2-signatures given in this section
can be combined with any other algebraic 2-signature, yielding an effective 2-signature thanks to Theorem~\ref{thm:alg-elem-2-sigs-are-representable}.

\subsection{Monoids}
\label{sec:example-monoids}

We begin with an example of monad for a first-order syntax with equations.
Given a set $X$, we denote by $M(X)$ the free
monoid built over $X$.  This is a classical example of monad over the
category of (small) sets.  The monoid structure gives us, for each set
$X$, two maps $m_X\colon M(X) \times M(X) \rar M(X)$ and $e_X\colon 1
\rar M(X)$ given by the product and the identity respectively.  It can
be easily verified that $m\colon M^2 \rar M$ and $e\colon 1 \rar M$
are $M$-module morphisms.  In other words, $(M,\rho) =
(M,[m, e])$ is a model of the 1-signature 
$\Sigma = \Theta\times\Theta + 1$.

We break the tautological morphism of $\Sigma$-modules (cf.\ Example~\ref{ex:sigma-mor-of-sig-mor})
into constituent pieces, defining $\m := \tau \circ \inl : \Theta \times \Theta \to \Theta$ and
$\e := \tau \circ \inr : 1 \to \Theta$.

Over the 1-signature $\Sigma$ we specify equations postulating
\emph{associativity} and \emph{left and right unitality}
as follows:
\begin{equation*}
  \xymatrix@C=25pt@R=2pt{
    \Theta^3 \ar[rr]^{\Theta\times \m} &&
    \Theta^2\ar[r]^{\m} & \Theta \\
    \Theta^3 \ar[rr]_{\m\times \Theta} &&
    \Theta^2\ar[r]_{\m} & \Theta
    }
    \qquad
  \xymatrix@C=25pt@R=2pt{
    \Theta\ar[r]^{\e\times \Theta} & \Theta^2 \ar[r]^{\m} & \Theta \\
    \Theta \ar[rr]_{1} && \Theta
    }
  \qquad
  \xymatrix@C=25pt@R=2pt{
    \Theta\ar[r]^{\Theta\times \e} & \Theta^2 \ar[r]^{\m} & \Theta \\
    \Theta\ar[rr]_{1} && \Theta
    }
\end{equation*}
and we denote by $E$ the family consisting of these three
$\Sigma$-equations. All are elementary since their codomain is $\Theta$,
and their domain a product of $\Theta$s.

One checks easily that $(M,[m,e])$ is the initial model of $(\Sigma,E)$.

Several other classical (equational) algebraic
theories, such as groups and rings, can be treated similarly, see Section~\ref{sec:alg-ths} below.
However, at the present state we cannot model
theories with partial construction (e.g., fields).

\subsection{Colimits of algebraic 2-signatures}
\label{s:col-alg-sigs}
In this section, we argue that our framework encompasses any
colimit of algebraic 2-signatures. 

Actually, the class of algebraic 2-signatures is not stable under colimits,
as this is not even the case for algebraic 1-signatures. However, we can
weaken this statement as follows:
\begin{proposition}\label{prop:colimit-of-algs}
Given any colimit of algebraic 2-signatures, there
is an algebraic 2-signature yielding an isomorphic category of models.
\end{proposition}
\begin{proof}
  As the class of algebraic 2-signatures is closed under arbitrary coproducts,
  using the decomposition of colimits into coproducts and coequalizers,
  any colimit $\Xi$ of algebraic 2-signatures can be expressed as a coequalizer
  of two morphisms $f,g$ between some
  algebraic 2-signatures $(\Sigma_1,E_1)$ and $(\Sigma_2,E_2)$,
  \[
\xymatrix{
 (\Sigma_1,E_1) \ar@<.5ex>[r]^f \ar@<-.5ex>[r]_g & (\Sigma_2,E_2) \ar[r]^p & 
 \Xi=(\Sigma_3,E_2)
}.
\]
where $\Sigma_3$ is the coequalizer of the 1-signatures morphisms $f$ and $g$.
Note that the set of equations of $\Xi$ is  $E_2$, by definition of the coequalizer in
the category of 2-signatures.
Now, consider the algebraic 2-signature $\Xi' = (\Sigma_2, E_2 + \eqref{eq:colim-of-algs})$ consisting of the 1-signature $\Sigma_2$
and the equations of $E_2$ plus the following elementary equation (see Example~\ref{ex:elementary-algebraic}):
\begin{equation} \label{eq:colim-of-algs}
      \xymatrix@R=2pt{
        \Sigma_1 \ar[r]^{f} &
        \Sigma_2 \ar[r]^{\tau^{\Sigma_2}} &
        \Theta \\
        \Sigma_1 \ar[r]_{g} &
        \Sigma_2 \ar[r]_{\tau^{\Sigma_2}} &
        \Theta 
      }
\end{equation}
We show that $\Mon^{\Xi}$ and $\Mon^{\Xi'}$ are isomorphic.
A model of $\Xi'$ is a monad $R$ together with an $R$-module morphism
 $r:\Sigma_2(R)\to R$ such that $r\circ f_R=r\circ g_R$ and that the equations of $E_2$ are satisfied. By universal property of the
 coequalizer, this is exactly the same as giving an $R$-module morphism
 $\Sigma_3(R)\to R$ satisfying the equations of $E_2$, i.e., 
 giving $R$ an action of $\Xi=(\Sigma_3,E_2)$.

 It is straightforward to check that this correspondence yields an isomorphism
 between the category of models of $\Xi$ and the category of models of $\Xi'$.
\end{proof}

This proposition, together with the following corollary, allow us 
to recover all the examples presented in \cite{ahrens_et_al:LIPIcs:2018:9671}, as
  colimits of algebraic 1-signatures:
syntactic commutative binary operator, maximum operator, application à la 
differential lambda calculus, syntactic closure operator, integrated substitution
operator, coherent fixpoint operator.
\begin{corollary}
  \label{c:finitary-colim}
  If $F$ is a finitary endofunctor on $\Set$, then there is an algebraic 2-signature whose category
  of models is isomorphic to the category of 1-models of the 1-signature $F
  \cdot \Theta$.
\end{corollary}
\begin{proof}
  It is enough to prove that $F\cdot\Theta$ is a colimit of algebraic 1-signatures.

  As $F$ is finitary, it is isomorphic to the coend $\int^{n\in\NN}F(n)
  \times \_^n$ 
  where $\NN$ is the full subcategory of $\Set$ of finite ordinals (see, e.g., \cite[Example~3.19]{DBLP:journals/mscs/VelebilK11}).
  As colimits are computed pointwise, the 1-signature $F\cdot \Theta$ is the coend $\int^{n\in
    \NN}F(n)\times \Theta^n$, and as such, it is a colimit of algebraic 2-signatures.
\end{proof}

However, we do not know whether we can recover our theorem \cite[Theorem~35]{ahrens_et_al:LIPIcs:2018:9671} stating 
  that any presentable 1-signature is effective.

\subsection{Algebraic theories}
\label{sec:alg-ths}
From the categorical point of view, several fundamental algebraic
structures in mathematics can be conveniently and elegantly described
using finitary monads.  For instance, the category of monoids can be
seen as the category of Eilenberg–Moore algebras of the monad of
lists.  Other important examples, like groups and rings, can
be treated analogously.  A classical reference on the subject is the
work of Manes, where such monads are significantly called
\emph{finitary algebraic theories} \cite[Def.\ 3.17]{Manes}.

We want to show that such \enquote{algebraic theories} fit in our
framework, in the sense that they can be incorporated into an algebraic
2-signature, with the effect of enriching the initial model with the operations
of the algebraic theory, subject to the axioms of the algebraic theory.

For a finitary monad $T$, Corollary~\ref{c:finitary-colim} 
says how to encode the 1-signature $T\cdot \Theta$ as an algebraic 2-signature $(\Sigma_T,E_T)$.
 Models are monads $R$ together with an $R$-linear morphism $r:T\cdot R \to R$.

Now, for any model $(R,m)$ of $T\cdot\Theta$, we would like to enforce the usual $T$-algebra equations on the action $m$.
This is done thanks to the following equations, where $\tau$ denotes the tautological morphism of $T\cdot\Theta$-modules:
\begin{equation}
\label{eq:algebra-eqs}
 \begin{aligned}
 \xymatrix@R=2pt{ \Theta \ar[r]^-{\eta_T \cdot \Theta} & T\cdot \Theta
                     \ar[r]^-{\tau} &  \Theta 
                     \\
                  \Theta \ar[rr]_{1} && \Theta }
\end{aligned}
\qquad\qquad
 \begin{aligned}
 \xymatrix@R=2pt{
          T \cdot T \cdot \Theta \ar[r]^-{\mu_T \cdot \Theta} & T \cdot \Theta
          \ar[r]^-{\tau} &
          \Theta 
          \\
          T \cdot T \cdot \Theta \ar[r]_-{T {\tau}} & T \cdot \Theta
          \ar[r]_-{\tau} & \Theta }
 \end{aligned}
\end{equation}
The first equation is clearly elementary.
The second one is elementary thanks to the following lemma:
\begin{lemma}\label{lem:finitary-preserves-epi}
  Let $F$ be a finitary endofunctor on $\Set$.  Then $F$ preserves epimorphisms.
\end{lemma}
\begin{proof}
As $F$ is finitary, it is isomorphic to the coend $\int^{n\in\NN}F(n) \times \_^n$  \cite[Example~3.19]{DBLP:journals/mscs/VelebilK11}.
By decomposing it as a coequalizer of coproducts, we get an epimorphism
$\alpha : \coprod_{n\in\NN}F(n)\times\_^n \to F$.
Now, let $f:X\to Y$ be a surjective function between two sets. We show
that $F(f)$ is epimorphic. By naturality, the following diagram commutes:
\[
  \xymatrix{
    \coprod_{n\in\NN}F(n)\times X^n \ar[r]^{F(n)\times f^n}
    \ar[d]_{\alpha_X}
    & 
    \coprod_{n\in\NN}F(n)\times Y^n \ar[d]^{\alpha_Y}
    \\
    F(X) \ar[r]_{F(f)} & F(Y)
  }
\]
The composition along the top-right is epimorphic by composition of epimorphisms. 
Thus, the bottom left is also epimorphic, and so is $F(f)$ as the last morphism
of this composition.
    
\end{proof}

In conclusion, we have exhibited the algebraic 2-signature $(\Sigma_T , E'_T)$,
where
$E'_T$ extends the family $E_T$  with the two elementary equations of
Diagram~\ref{eq:algebra-eqs}. This signature
 allows to enrich any other algebraic 2-signature with the operations of the
algebraic theory $T$, subject to the relevant equations.

\subsection{Fixpoint operator}\label{sec:fixpoint-operator}

Here, we show the algebraic 2-signature corresponding to a fixpoint operator. In
\cite[Section~9.4]{ahrens_et_al:LIPIcs:2018:9671} we studied 
fixpoint operators in the context of 1-signatures.  In that setting,
we treated a \emph{syntactic} fixpoint operator
called \emph{coherent} fixpoint operator, somehow reminiscent of mutual letrec. We were able to impose many natural equations to this operator but we were
not able to enforce the fixpoint equation.  In this
section, we show how a fixpoint operator can be fully specified by an algebraic
2-signature. We restrict our discussion to the unary
case; 
the coherent family of multi-ary fixpoint operators 
presented in \cite[Section~9.4]{ahrens_et_al:LIPIcs:2018:9671},
now including the fixpoint equations,
can also be specified, in an analogous way,
via an algebraic 2-signature.

Let us start by recalling the following

\begin{definition}
\label{def:unary-fp-op}
A \textbf{unary fixpoint operator for a monad $R$}
\cite[Definition~50]{ahrens_et_al:LIPIcs:2018:9671} is a module
morphism $f$ from $R'$ to $R$ that makes the following diagram 
commute, where $\sigma$ is the substitution morphism
defined as the uncurrying (see Definition~\ref{def:half-eq-curry}) of the identity morphism on $\Theta'$:
\begin{equation*}
  \label{eq:fixpoint}
    \begin{tikzcd}
      R'
      \arrow[rr, "(id_{R'} \text{,} f)"]
      \arrow[rd,"f",swap]
      & & R' \times R \arrow[ld, "\sigma_R"] \\
      & R
    \end{tikzcd}
\end{equation*}
  \end{definition}
  
  In order to rephrase this definition, we introduce the obviously algebraic 2-signature $\Upsilon_\fix$
  consisting of the
1-signature $\Sigma_\fix = \Theta'$ and the
family $E_\fix$ consisting of the single following $\Sigma_\fix$-equation:
\begin{equation}
\label{diag:fixpoint}
  e_\fix :
    \vcenter{\vbox{
  \xymatrix@C=35pt@R=2pt{
    \Theta' \ar[r]^-{\langle 1,\tau \rangle} &
    \Theta'\times\Theta \ar[r]^-{\sigma} & \Theta \\
    \Theta' \ar[rr]_{\tau} && \Theta
    }}}
\end{equation}

This allows us to rephrase the previous definition as follows: a unary fixpoint operator for a monad $R$ is just an action of the 2-signature $\Upsilon_\fix$ in $R$.

  The name \enquote{fixpoint operator} is motivated by the following proposition:
\begin{proposition}
  [{\cite[Proposition~51]{ahrens_et_al:LIPIcs:2018:9671}}]
  \label{prop:fixpoint-combinator-action}
  Fixpoint combinators are in one-to-one correspondence with
  actions of $\Upsilon_\fix$ in the 
 monad $\LCb$ of the lambda calculus modulo $\beta$- and $\eta$-equality.
\end{proposition}
Recall that fixpoint combinators
are lambda terms $Y$ 
satisfying, for any (possibly open) term $t$, the equation
\[
  \app(t, \app(Y, t)) = \app(Y,t) \enspace .
\]
Explicitly, such a combinator $Y$ induces a fixpoint operator
$\hat{Y}:\LCb'\rightarrow \LCb$ which associates, to any term $t$
depending on an additional variable $*$, the term
$\hat Y (t) := \app(Y,\abs \: \, t)$.

\section{Recursion}
\label{sec:recursion}

In this section, we explain how a recursion principle can be derived from our initiality result, and
give an example of a morphism---a \emph{translation}---between
monads defined via the recursion principle.

\subsection{Principle of recursion}
\label{ss:rec}

In our context, the recursion principle is a recipe for constructing a morphism
from the monad underlying the initial model of a 2-signature to an arbitrary monad.

\begin{proposition}[Recursion principle]
  Let $S$ be the monad underlying the initial model of the 2-signature $\Upsilon$. To any action $a$ of $\Upsilon$ in $T$ is associated a monad morphism $\hat a : S \to T$.  
\end{proposition}
\begin{proof}
The action $a$ defines a 2-model $M$ of $\Upsilon$, and $\hat a$ is the monad morphism underlying the initial morphism to $M$. 
\end{proof}

Hence the recipe consists in the following two steps:
\begin{enumerate}
\item give $T$ an action of the 1-signature $\Sigma$;
 \item check that all the equations in $E$ are satisfied for the induced model.
\end{enumerate}

In the next section, we illustrate this principle.

\subsection{Translation of lambda calculus with fixpoint to lambda calculus}
\label{sec:translate-rec}

In this section, we consider the 2-signature
 $\Upsilon_{\LCbfix} := \Upsilon_{\LCb} + \Upsilon_{\fix}$ where the two components have been introduced above  (see Example~\ref{ex:2sig-lcbeta} and Section~\ref{sec:fixpoint-operator}).
 
 As a coproduct of algebraic 2-signatures, $\Upsilon_{\LCbfix}$ is itself
algebraic, and thus the initial model exists.
 The underlying monad $\LCbfix$ of the initial model can be understood as the monad 
of lambda calculus modulo $\beta$ and $\eta$ enriched with an \emph{explicit} fixpoint
 operator $\fix:\LCbfix'\rar \LCbfix$.
Now we build by recursion a monad morphism from this monad to the \enquote{bare}
 monad $\LCb$ of lambda calculus modulo $\beta$ and $\eta$.

As explained in Section~\ref{ss:rec}, we need to define an action of $\Upsilon_{\LCbfix}$ in  $\LCb$, that is to say an action of
$\Upsilon_{\LCb}$ plus an action of $\Upsilon_{\fix}$.
For the action of $\Upsilon_{\LCb}$, we take the one yielding the initial model. 

Now, in order to find an action of $\Upsilon_{\fix}$ in $\LCb$, we choose a fixpoint combinator $Y$ (say the one of Curry) and take the action $\hat Y$ as defined at the end of Section~\ref{sec:fixpoint-operator}.

In more concrete terms, our translation is a kind of compilation which replaces each occurrence of the explicit fixpoint operator $\fix(t)$
with $\app(Y,\abs \: \, t)$.

\bibliographystyle{plainurl}
\bibliography{strengthened}


\end{document}